\documentclass[lettersize,onecolumn]{IEEEtran}
\usepackage{mathrsfs}
\usepackage{amsmath,amsfonts}
\usepackage{amssymb,amscd,amsthm}
\usepackage{algorithmic}
\usepackage{algorithm}
\usepackage{array}
\usepackage[caption=false,font=normalsize,labelfont=sf,textfont=sf]{subfig}
\usepackage{textcomp}
\usepackage{stfloats}
\usepackage{url}
\usepackage{verbatim}
\usepackage{graphicx}
\usepackage{cite}
\newtheorem{lemma}{\qquad\textbf{Lemma}}
\newtheorem{theorem}{\qquad\textbf{Theorem}}
\newtheorem{definition}{\qquad\textbf{Definition}}
\newtheorem{remark}{\qquad\textbf{Remark}}
\newtheorem{corollary}{\qquad\textbf{Corollary}}
\newtheorem{example}{\qquad\textbf{Example}}
\newtheorem{proposition}{\qquad\textbf{Proposition}}

\begin{document}

\title{\bf Column Twisted Reed-Solomon Codes as MDS Codes}

\author{Wei Liu,~Jinquan Luo,~Puyin Wang,~and Dengxin Zhai
\thanks{The authors are with School of Mathematics and Statistics \& Hubei Key Laboratory of Mathematical Sciences, Central China Normal University, Wuhan China 430079.  Dengxin Zhai is also with the School of Mathematics and Statistics, Kashi University, Kashi 844000, China.  The research  is supported by National Natural Science Foundation of China (Nos.12441102, 12171191, 12271199), Natural Science Foundation of Xinjiang Uygur Autonomous Region (2022D01B128),  SRMC Fund (No.2024SRMC01) and the Fundamental Research Funds for the Central Universities (No.CCNU25JCPT031).}
\thanks{E-mail: 1450820784@qq.com(W.Liu),  luojinquan@mail.ccnu.edu.cn(J.Luo), 453798449@qq.com(P.Wang), dxzhai2022@126.com(D.Zhai)}}

\markboth{Journal of \LaTeX\ Class Files,~Vol.~1, No.~2, December~2023}%
{Shell \MakeLowercase{\textit{et al.}}: A Sample Article Using IEEEtran.cls for IEEE Journals}

\IEEEpubid{0000--0000~\copyright~2023 IEEE}

\maketitle

\begin{abstract}
In this paper, we study column twisted Reed-Solomon(TRS) codes. We establish some sufficient conditions for these codes to be MDS  and show that the dimension of their Schur square codes is $2k$. Consequently, these TRS codes are shown to be not equivalent to Reed-Solomon(RS) codes. Moreover, our construction offers more flexible parameters than existing twisted generalized Reed-Solomon(TGRS) code designs. For a large odd prime power $q$, systematically constructed TGRS codes are known to be limited to length  $\frac{q+1}{2}$.  By contrast, our column TRS construction supports code lengths up to $\frac{q+3}{2}$. Finally, we present the dual codes of column TRS codes. Overall, this work introduces a new method for constructing MDS codes by appending column vectors to some generator matrix of an RS code.
\end{abstract}


\section{Introduction}
\IEEEPARstart{L}{et} $\mathbb{F}_q$ be a finite field of size $q$, where $q$ is a prime power and $\mathbb{F}_q^* = \mathbb{F}_q \setminus \{0\}$. An $[n,k,d]$ linear code $C$ is a subspace of the linear space $\mathbb{F}_q^n$ over $\mathbb{F}_q$ with dimension $k$ and minimum Hamming distance $d$. It is well known that the parameters of $C$ satisfy $d \leq n - k + 1$. If $d = n - k + 1$, $C$ is called \emph{maximum distance separable}(MDS). Many researchers have studied MDS codes(\!\!\cite{CN2018}, \cite{GA2011}), their covering radius(\!\!\cite{BO2015}), weight distribution(\!\!\cite{ET2011}),   LCD property(\!\!\cite{CE2018}), self-dual property(\!\!\cite{GN2008}, \cite{LC2019}), classification(\!\!\cite{KO2015}), and related cryptographic topics(\!\!\cite{CO2008}). \emph{Generalized Reed-Solomon}(GRS) codes are an important class of MDS linear codes, which can correct burst errors and have many applications in different occasions, such as  providing  high fidelity in CD player. Furthermore, many researchers have investigated non-GRS MDS codes. For instance, certain arcs in finite geometry correspond to MDS codes that are non-equivalent to GRS codes. Let $\mathcal{A}$ be the set of points obtained from the columns of a $k \times n$ generator matrix of a linear MDS code. If the code is a Reed-Solomon code then $\mathcal{A}$ is a rational normal curve. The MDS property implies that any $k$ points of $\mathcal{A}$ span the whole space. A set of points of $\mathrm{PG}(k-1,q)$ with this property is called an arc $\mathcal{A}$. For $k = 3$ and $q$ even, many such examples are known, for instance,
$\mathcal{A} = \{(1,x,x^{\sigma}) \mid x \in \mathbb{F}_q\} \cup \{(0,1,0), (0,0,1)\}$, for certain automorphisms $\sigma$ of $\mathbb{F}_q$(\!\!\cite{CA2003}, \cite{PI1995}). In \cite{ZN2025}, Zhi et al. presented some non-GRS MDS and NMDS codes.

From now on T is referred to as twisted, G as generalized, and RS as Reed-Solomon (code). \emph{Twisted Reed-Solomon}(TRS) codes, which are generalizations of RS codes, were firstly introduced in \cite{BT2017}. Unlike RS codes, TRS codes may not be MDS codes. Therefore, many researchers are interested in the conditions under which TRS codes are MDS codes, non-equivalent to MDS codes. Generally, there are two types of
TGRS codes: single-twist and multi-twists. A TGRS code with single-twist is obtained by adding one monomial to each code polynomial, and the one with multi-twists is obtained by adding $\ell$ monomials to each code polynomial ($\ell \geq 2$).

For TGRS codes with single-twist, based on $(t, h) = (1, 0)$,   Beelen et al. gave  the necessary and sufficient conditions for TGRS codes to be MDS codes, and concluded that their code length satisfies $n \leq \frac{q + 1}{2}$(\!\cite{BT2017}). Zhang et al. investigated the case of TGRS codes when $(t, h) = (q - k - 1, k - 1)$ and provided the necessary and sufficient conditions for them to be  self-dual (\!\!\cite{ZA2022}). Subsequently, Huang et al. and Sui et al. characterized the necessary and sufficient
condition that a TGRS code is MDS for any position pair $(t,h)$ in \cite{HM2023},\cite{HM2021},\cite{SJM2022}. For the condition under which a TGRS code with multi-twists is MDS, Sui et al. characterized the
 necessary and sufficient conditions for a TGRS code
constructed by adding two polynomials both of which contain
no more than two monomials to each code polynomial(\!\!\cite{SM2022}). Gu et al. provided the necessary and sufficient condition for $l$-twist TGRS codes constructed by adding $l$ special monomials to each code polynomial(\!\!\cite{GO2024}). The construction in \cite{SM2024} is similar to the approach in \cite{GO2024}, differing in that \cite{GO2024} adds $l$ arbitrary monomials to each code polynomial. In the aforementioned studies, the components of hooks in TGRS codes were all unequal. Subsequently in \cite{CA2024}, Zhu and Liao established the necessary and sufficient conditions for TGRS codes to be MDS  when $\boldsymbol{t} = (1, 2)$ and $\boldsymbol{h} = (k-2, k-2)$.  TGRS code for arbitrary position pairs of $(t,h)$  is called an arbitrary twists Reed-Solomon(A-TGRS) code. Building upon this foundation, Zhao et al. established the necessary and sufficient conditions for A-TGRS codes to be MDS in more general case, which essentially encompasses all previously studied special cases(\!\!\cite{ZR2025}). In summary, their studies exclusively focused on row transformations of the generator matrix of GRS codes. Consequently, all these codes could only achieve code length $n \leq \frac{q+1}{2}$. This limitation motivated our investigation into column transformation of  GRS generator matrix. Unlike previous approaches,  these non Reed-Solomon type MDS codes in our construction can attain a maximum code length of $\frac{q+3}{2}$.

There are many other works besides the MDS conditions of TGRS codes, such as self-dual codes in \cite{YN2025},\cite{GD2023}, \cite{GO2024},\cite{ZA2022}, small hulls in \cite{SM2024},\cite{WT2021},  parity-check matrices in \cite{CO2024}, deep holes in \cite{FD2024},\cite{LC2025}, cryptanalysis in \cite{BS2018},\cite{LC2020} and decoding algorithms in \cite{SD2025}. There are also some discussions on (+)-twist in \cite{CT2024}.

This paper is organized as follows. In Section \uppercase\expandafter{\romannumeral2}, we show some basic notations and results about TGRS codes. In Section \uppercase\expandafter{\romannumeral3}, we characterize the necessary and sufficient conditions for column Twisted Reed-Solomon(column TRS) codes to satisfy  MDS property, while demonstrating their non-equivalence to conventional (extended) GRS codes. In Section \uppercase\expandafter{\romannumeral4}, we determine parity-check matrix of column TGRS codes, thereby obtaining their dual codes. In Section \uppercase\expandafter{\romannumeral5}, we conclude our work.

\section{Preliminaries}

Throughout this paper, $\mathbb{F}_{q}$ denotes a finite field of order $q$, where $q$ is a prime power and $\mathbb{F}_{q}^{*}=\mathbb{F}_{q}\setminus\{0\}$. Let $\mathbb{F}_{q}[x]$ be the polynomial ring over $\mathbb{F}_{q}$. Given a vector
\[\boldsymbol{a}=\left( {{a_1},{a_2}, \ldots ,{a_n}} \right) \in \mathbb{F}_q^n,\]
where $a_1,a_2,\ldots,a_n$ are distinct elements of $\mathbb{F}_q$.  Usually, $a_1,a_2,\ldots,a_n$ are called $\mathit{evaluation\ points}$. Moreover, given another vector
\[
\boldsymbol{v}=(v_1,v_2,\ldots,v_n)\in(\mathbb{F}_q^*)^n,
\]
a GRS code is defined as
\[
\begin{split}
\mathcal{GRS}_{k,n}(\boldsymbol{a},\boldsymbol{v}) = \left\{
\left(v_1 f(a_1), \ldots, v_n f(a_n)\right) \mid f(x) \in \mathbb{F}_q[x],\ \deg f(x) \leq k-1
\right\}.
\end{split}
\]
An extended GRS code is defined as
\[
\begin{split}
\mathcal{GRS}_{k,n}(\boldsymbol{a},\boldsymbol{v},\infty)
&= \left\{ \left(v_1 f(a_1), \ldots, v_n f(a_n), f_{k-1}\right) \mid f(x) \in \mathbb{F}_q[x],\ \deg f(x) \leq k-1 \right\}
\end{split},
\]
where $f_{k-1}$ is the coefficient of $x^{k-1}$ in $f(x)$.
If $\boldsymbol{v} = \mathbf{1}$, then $\mathcal{GRS}_{k,n}(\boldsymbol{a},\mathbf{1})$ and
$\mathcal{GRS}_{k,n}(\boldsymbol{a},\mathbf{1},\infty)$ are the RS code and the extended RS code, respectively.  Linear space  $\mathcal{V}_{k,\boldsymbol{t},\boldsymbol{h},\boldsymbol{\eta}}$ consisting of twisted polynomials is given in the following.

\begin{definition}\cite{BT2017}
For two positive integers $l,k$ with $l \leq k \leq n \leq q$, suppose that:
\begin{itemize}
    \item $\boldsymbol{h} = (h_1, h_2, \ldots, h_l)$, where $0 \leq h_i \leq k - 1$ are distinct,
    \item $\boldsymbol{t} = (t_1, t_2, \ldots, t_l)$, where $0 \leq t_i < n - k$ are also distinct,
    \item $\boldsymbol{\eta} = (\eta_1, \eta_2, \ldots, \eta_l) \in \mathbb{F}_q^l$.
\end{itemize}
Then
\[
\begin{split}
\mathcal{V}_{k,\boldsymbol{t},\boldsymbol{h},\boldsymbol{\eta}}
&= \left\{ f(x) = \sum_{i=0}^{k-1} f_{i}x^{i} + \sum_{j=1}^{l}\eta_{j}f_{h_{j}}x^{k+t_{j}} \middle|\; f_i\in\mathbb{F}_q \ (i = 0,\ldots,k-1)  \right\}
\end{split}
\]
is a $k$-dimensional $\mathbb{F}_q$-linear subspace of $\mathbb{F}_q[x]$. Here $\boldsymbol{h}$ and $\boldsymbol{t}$ are called the $\mathit{hooks}$ and the $\mathit{twists}$, respectively.
\end{definition}
Then a TGRS code is defined as
\[ \small
\mathcal{C}_{k,n,\boldsymbol{t},\boldsymbol{h}}(\boldsymbol{a},\boldsymbol{v},\boldsymbol{\eta}) = \left\{ (v_1f(a_1),\ldots,v_nf(a_n)) \mid f(x)\in\mathcal{V}_{k,\boldsymbol{t},\boldsymbol{h},\boldsymbol{\eta}} \right\}.
\]
For $l > 0$, an extended TGRS code is defined as
\[
\mathcal{C}_{k,n,\boldsymbol{t},\boldsymbol{h}}(\boldsymbol{a},\boldsymbol{v},\boldsymbol{\eta},\infty) =
\left\{ (v_1f(a_1),\ldots,v_nf(a_n),f_l) \mid f(x)\in\mathcal{V}_{k,\boldsymbol{t},\boldsymbol{h},\boldsymbol{\eta}} \right\},
\]
where $f_l$ is the coefficient of $x^l$ in $f(x)$.
In particular, $\mathcal{C}_{k,n,\boldsymbol{t},\boldsymbol{h}}(\boldsymbol{a},\mathbf{1},\boldsymbol{\eta})$ and $\mathcal{C}_{k,n,\boldsymbol{t},\boldsymbol{h}}(\boldsymbol{a},\mathbf{1},\boldsymbol{\eta},\infty)$ are called  TRS code and extended TRS code respectively. The Schur product of linear codes is defined as follows.
\begin{definition}\label{schur product}
Let $\boldsymbol{x}=(x_1,x_2,\cdots ,x_n)$, $\boldsymbol{y}=(y_1,y_2,\cdots ,y_n) \in {\left( {\mathbb{F}_q^ * } \right)^n}$, the \textit{Schur product} of $\boldsymbol{x}$ and $\boldsymbol{y}$ is defined as $\boldsymbol{x}*\boldsymbol{y}:=\left( {{x_1}{y_1}, \ldots ,{x_n}{y_n}} \right)$. The Schur product of two linear codes $\mathcal{C}_1,\mathcal{C}_2 \subseteq \mathbb{F}_q^n$ is defined as
\[\mathcal{C}_1*\mathcal{C}_2: = {\left\langle {{\textit{\textbf{c}}_1}*{\textit{\textbf{c}}_2}:{\textit{\textbf{c}}_1} \in \mathcal{C}_1,{\textit{\textbf{c}}_2} \in \mathcal{C}_2} \right\rangle _{{\mathbb{F}_q}}},\]
where ${\left\langle S \right\rangle _{{\mathbb{F}_q}}}$ denotes the $\mathbb{F}_q$-linear subspace generated by the subset $S$ of $\mathbb{F}_q^n$.
\end{definition}
In particular, if $\mathcal{C}_1=\mathcal{C}_2=\mathcal{C}$, we call $\mathcal{C}^2:=\mathcal{C}*\mathcal{C}$ the Schur square code of $\mathcal{C}$. The definition of the equivalence for linear codes is given in the following.
\begin{definition}\cite{BT2017}
  Let $\mathcal{C}$,~$\mathcal{D}$ be $[n,k]$ linear codes over $\mathbb{F}_q$. We say that $\mathcal{C}$ and $\mathcal{D}$ are \textit{equivalent} if there is a permutation $\pi \in S_n$ and $\boldsymbol{v}=(v_1,v_2,\cdots ,v_n)\in {\left( {\mathbb{F}_q^ * } \right)^n}$ such that $\mathcal{C}={\varphi _{\pi ,\textit{\textbf{v}}}}(\mathcal{D})$ where ${\varphi _{\pi ,\boldsymbol{v}}}$ is the Hamming-metric isometry

\[{\varphi _{\pi ,v}}:\mathbb{F}_q^n \to \mathbb{F}_q^n,\left( {{c_1}, \ldots ,{c_n}} \right) \mapsto \left( {{v_1}{c_{\pi \left( 1 \right)}}, \ldots ,{v_n}{c_{\pi \left( n \right)}}} \right).\]
\end{definition}

\begin{remark}
It is easy to see that $\mathcal{C}_{1}^{2}$ and $\mathcal{C}_{2}^{2}$ are equivalent when $\mathcal{C}_{1}$ and $\mathcal{C}_{2}$ are equivalent.
\end{remark}
By the definitions of the GRS code and the extended GRS code, and Definition \ref{schur product}, we can get the following proposition directly.
\begin{proposition}\cite{CD2014}
  If $k \leq \frac{n}{2}$, then
\[
\mathcal{GRS}_{k,n}^{2}(\boldsymbol{a},\mathbf{1}) = \mathcal{GRS}_{2k-1,n}(\boldsymbol{a},\mathbf{1})
\]
and
\[
\mathcal{GRS}_{k,n-1}^{2}(\boldsymbol{a},\mathbf{1},\infty) = \mathcal{GRS}_{2k-1,n-1}(\boldsymbol{a},\mathbf{1},\infty).
\]
\end{proposition}
Hence, if an $[n,k]$ code $\mathcal{C}$ with $k<\frac{n+1}{2}$ satisfies $\dim(\mathcal{C}^2)\neq 2k-1$, then it is not equivalent to GRS code. In this paper, we also verify the non-GRS property of column TRS codes by the Schur product.

\section{Column Twisted Reed-Solomon Codes}
In this section, we first investigate one-column TRS codes. Building on this foundation, we provide some necessary and sufficient conditions for two-column TRS codes to be MDS, and verify their non-GRS property. Below we first give the definition of one-column TGRS codes.

For distinct $a_i, b, c\in \mathbb{F}_q$,  define
\[G_1(b,c)=\left(
      \begin{array}{ccccc}
        1 & 1 & \cdots & 1 & 1-\lambda  \\
        a_1 & a_2 & \cdots & a_{n-1}  & b-\lambda c  \\
        a_1^2 & a_2^2 & \cdots & a_{n-1}^2 & b^2-\lambda c^2  \\
        \vdots & \vdots & \vdots & \vdots & \vdots  \\
        a_1^{k-1} & a_2^{k-1} & \cdots & a_{n-1}^{k-1} & b^{k-1}-\lambda c^{k-1} \\
      \end{array}
    \right)
\] and
\[
G_1(b,c,\infty) =
\left(
      \begin{array}{cccccc}
        1 & 1 & \cdots & 1 & 1-\lambda & 0 \\
        a_1 & a_2 & \cdots & a_{n-1}  & b-\lambda c&  0  \\
        a_1^2 & a_2^2 & \cdots & a_{n-1}^2 & b^2-\lambda c^2&  0  \\
        \vdots & \vdots & \vdots & \vdots & \vdots & \vdots \\
        a_1^{k-2} & a_2^{k-2} & \cdots & a_{n-1}^{k-2} & b^{k-2}-\lambda c^{k-2}& 0 \\
        a_1^{k-1} & a_2^{k-1} & \cdots & a_{n-1}^{k-1} & b^{k-1}-\lambda c^{k-1}& 1  \\
      \end{array}
    \right).
\]
Denote by $C_1(b,c)$ and $C_1(b,c,\infty)$  linear codes generated by $G_1(b,c)$ and $G_1(b,c,\infty)$ respectively.
\begin{lemma}\label{one column criterion}
\begin{itemize}
  \item[(1)] The code $C_1(b,c)$   is MDS if and only if
 for any $1\leq i_1\leq \cdots \leq i_{k-1}\leq n-1$,
 \[\prod\limits_{j=1}^{k-1}(b-a_{i_j})-\lambda\prod\limits_{j=1}^{k-1}(c-a_{i_j})\neq 0.\]
  \item[(2)] The code $C_1(b,c,\infty)$   is MDS if and only if
 for any $1\leq i_1\leq \cdots \leq i_{s}\leq n-1$ with $s\in \{k-2, k-1\}$,
 \[\prod\limits_{j=1}^{s}(b-a_{i_j})-\lambda\prod\limits_{j=1}^{s}(c-a_{i_j})\neq 0.\]
\end{itemize}

\end{lemma}
\begin{proof} We only prove (1). The case (2) can be proven in a similar way and we omit the details here.

It is  sufficient to show that for  any $1\leq i_1\leq \cdots \leq i_{k-1}\leq n-1$, the submatrix of $G_1(b,c)$ consisting of the $i_1$-th, $\cdots$, $i_{k-1}$-th and $n$-th columns is nonsingular, that is,
\begin{multline*}
\det\left(
      \begin{array}{ccccc}
        1 & 1 & \cdots & 1& 1-\lambda \\
        a_{i_1} & a_{i_2} & \cdots &a_{i_{k-1}}& b-\lambda c \\
        \vdots & \vdots & \vdots & \vdots & \vdots\\
        a_{i_1}^{k-1} & a_{i_2}^{k-1} & \cdots & a_{i_{k-1}}^{k-1}& b^{k-1}-\lambda c^{k-1} \\
      \end{array}
    \right)
= \prod\limits_{1\leq j<l\leq k-1}(a_{i_l}-a_{i_j})\cdot \left(\prod\limits_{j=1}^{k-1}(b-a_{i_j})-\lambda\prod\limits_{j=1}^{k-1}(c-a_{i_j})\right)\neq 0.
\end{multline*}
Hence every $k$ columns of $G_1(b,c)$ are linearly independent which is equivalent to saying that $C_1(b,c)$ is MDS.
\end{proof}

Let $H$ be a  multiplicative subgroup of $\mathbb{F}_q^*$ with order not less than $n-1$.
\begin{theorem}\label{one column construction}
 For $1 \leq i \leq n-1$, let $a_i=(b-\mu_i c)/(1-\mu_i)$ where $\mu_i \in H\setminus\{1\}$ are distinct elements and $b \neq c$. If $\lambda\notin H$, then both $C_1(b,c)$ and $C_1(b,c,\infty)$ are MDS codes with parameters $[n,k,n-k+1]$ and $[n+1, k, n-k+2]$, respectively. In particular, if $3\leq k \leq \frac{n}{2}  $ and $\lambda\neq 0$, neither $C_1(b,c)$ nor $C_1(b,c,\infty)$ is  equivalent to any RS code or extended RS code.
\end{theorem}
\begin{proof}
Now we only show the result of $C_1(b,c)$. The code $C_1(b,c,\infty)$ can be proved similarly and we omit the details here.

First, we prove that the $a_i$ are pairwise distinct. Observe that $(a_i - b) = (a_i - c)\mu_i$. If $a_i = c$, then we would have $a_i - b = 0$ implying $a_i = b = c$, which is a contradiction. Therefore $a_i \neq c$, and similarly $a_i \neq b$.
Now suppose there exist $s \neq t$ with $1 \leq s , t \leq n-1$
such that $a_s = a_t$. Then we have:
$$
\frac{a_s - b}{a_s - c} = \frac{a_t - b}{a_t - c}
$$
which implies $\mu_s = \mu_t$, again a contradiction.
Hence all $a_i$ are distinct. Thus, for any $1\leq i_1\leq \cdots \leq i_{k-1}\leq n-1$,
\[\prod\limits_{j=1}^{k-1}\frac{b-a_{i_j}}{c-a_{i_j}}=\prod\limits_{j=1}^{k-1}\mu_{i_j}\in H.\]
Since $\lambda\notin H$, by Lemma \ref{one column criterion}, the code $C_1(b,c)$ is MDS.

It remains to show $C_1(b,c)$ for $\lambda\neq 0$ is not equivalent to any RS code or extended RS code.
 The Schur product of $C_1(b,c)$, denoted by $C_1(b,c)* C_1(b,c)$,  has a generator matrix whose rows are
 \[\left(a_1^{i+j}, a_2^{i+j},\cdots, a_{n-1}^{i+j}, (b^i-\lambda c^i)(b^j-\lambda c^j)\right).\]

 Note that $b\neq c$. Choosing $(i,j)=(0,2), (1,1)$ respectively, we obtain two rows in $C_1(b,c)* C_1(b,c)$:
   \[\left(a_1^{2}, a_2^{2},\cdots, a_{n-1}^{2}, (1-\lambda )(b^2-\lambda c^2)\right)\]
   and
   \[\left(a_1^{2}, a_2^{2},\cdots, a_{n-1}^{2}, (b-\lambda c)^2\right).\]
  The difference of the above two rows is
   \[\left(0,\cdots, 0, \lambda(b-c)^2\right)\neq {\bf 0}.\]
   Therefore, $C_1(b,c)* C_1(b,c)$ has a generator matrix
   \[ \small
    \left(
     \begin{array}{ccccc}
       1 & 1 & \cdots & 1 & \ast \\[2mm]
       a_1 & a_2 & \cdots & a_{n-1} & \ast \\[2mm]
       a_1^2 & a_2^2 & \cdots & a_{n-1}^2 & (1-\lambda )(b^2-\lambda c^2) \\
       \vdots & \vdots & \vdots & \vdots & \vdots \\
       a_1^{2k-2} & a_2^{2k-2} & \cdots & a_{n-1}^{2k-2} & \ast \\[2mm]
       a_1^2 & a_2^2 & \cdots & a_{n-1}^2 & (b-\lambda c)^2 \\[2mm]
       \vdots & \vdots & \vdots & \vdots & \vdots \\
     \end{array}
   \right)\]
  which is row equivalent to
    \[\small
    \left(
     \begin{array}{ccccc}
       1 & 1 & \cdots & 1 & \ast \\[2mm]
       a_1 & a_2 & \cdots & a_{n-1} & \ast \\[2mm]
       a_1^2 & a_2^2 & \cdots & a_{n-1}^2 & (1-\lambda )(b^2-\lambda c^2) \\[2mm]
       \vdots & \vdots & \vdots & \vdots & \vdots \\
       a_1^{2k-2} & a_2^{2k-2} & \cdots & a_{n-1}^{2k-2} & \ast \\[2mm]
       0 & 0 & \cdots & 0 & \lambda(b-c)^2 \\
       \vdots & \vdots & \vdots & \vdots & \vdots \\
     \end{array}
   \right)
   \]
   whose rank is  $2k$.  As a result, $C_1(b,c)$ is not equivalent to any $k$-dimensional (extended) RS code whose rank of Schur product is $2k-1$.
\end{proof}

Based on the above observations, we present the necessary and sufficient conditions for two-column TGRS codes to be MDS codes.

Suppose $a_i$ ($1\leq i\leq n-1$), $b$,  and $c$ in $\mathbb{F}_q$ are all distinct. For $\lambda_1\neq \lambda_2$,  define
\[
G_2(b,c) =
\left(
      \begin{array}{ccccc}
        1  & \cdots & 1 & 1-\lambda_1 & 1-\lambda_2 \\
        a_1  & \cdots & a_{n-1}  & b-\lambda_1 c  & b-\lambda_2 c \\
        a_1^2  & \cdots & a_{n-1}^2 & b^2-\lambda_1 c^2  & b^2-\lambda_2 c^2\\
        \vdots  & \vdots & \vdots & \vdots &\vdots \\
        a_1^{k-1}  & \cdots & a_{n-1}^{k-1} & b^{k-1}-\lambda_1 c^{k-1} & b^{k-1}-\lambda_2  c^{k-1}\\
      \end{array}
    \right)
\]
and
\[
G_2(b,c,\infty) =
\left(
      \begin{array}{cccccc}
        1  & \cdots & 1 & 1-\lambda_1 & 1-\lambda_2 & 0 \\
        a_1  & \cdots & a_{n-1}  & b-\lambda_1 c  & b-\lambda_2 c & 0 \\
        a_1^2  & \cdots & a_{n-1}^2 & b^2-\lambda_1 c^2  & b^2-\lambda_2 c^2 & 0\\
        \vdots  & \vdots & \vdots & \vdots &\vdots &\vdots \\
        a_1^{k-2}  & \cdots & a_{n-1}^{k-2} & b^{k-2}-\lambda_1 c^{k-2} & b^{k-2}-\lambda_2  c^{k-2} & 0\\
        a_1^{k-1}  & \cdots & a_{n-1}^{k-1} & b^{k-1}-\lambda_1 c^{k-1} & b^{k-1}-\lambda_2  c^{k-1} & 1\\
      \end{array}
    \right).
\]
\begin{lemma}\label{two column criterion}
  \begin{itemize}
\item[(1)] The code $C_2(b,c)$ generated by $G_2(b,c)$ is MDS if and only if for any $1\leq i_1\leq \cdots \leq i_{k-1}\leq n-1$,
\begin{equation}
\begin{split}
\prod\limits_{j=1}^{k-1}(b-a_{i_j}) - \lambda_1\prod\limits_{j=1}^{k-1}(c-a_{i_j}) &\neq 0, \\
\prod\limits_{j=1}^{k-1}(b-a_{i_j}) - \lambda_2\prod\limits_{j=1}^{k-1}(c-a_{i_j}) &\neq 0.
\end{split}
\label{two column}
\end{equation}
\item[(2)] Similarly, the code $C_2(b,c,\infty)$ generated by $G_2(b,c,\infty)$ is MDS if and only if for any $1\leq i_1\leq \cdots \leq i_{s}\leq n-1$, $s=k-2$ or $k-1$,
\begin{equation}
\begin{split}
\prod\limits_{j=1}^{s}(b-a_{i_j})-\lambda_1\prod\limits_{j=1}^{s}(c-a_{i_j}) &\neq 0, \\
\prod\limits_{j=1}^{s}(b-a_{i_j})-\lambda_2\prod\limits_{j=1}^{s}(c-a_{i_j}) &\neq 0.
\end{split}
\label{two column extended}
\end{equation}
\end{itemize}
\end{lemma}
\begin{proof}
  We only prove (1). The case (2) can be proven in a similar way and we omit it.

  The proof proceeds similarly to Lemma \ref{one column criterion}. The only difference lies in the need to verify the cases for $1\leq i_1\leq \cdots \leq i_{k-2}\leq n-1$,
\[D:=\left|\begin{array}{ccccc}
1 & \cdots & 1 & 1-\lambda_{1} & 1-\lambda_{2} \\
a_{i_{1}} & \cdots & a_{i_{k-2}} & b-\lambda_{1}c & b-\lambda_{2}c \\
\vdots &  & \vdots & \vdots & \vdots \\
a_{i_{1}}^{k-1} & \cdots & a_{i_{k-2}}^{k-1} & b^{k-1}-\lambda_{1}c^{k-1} & b^{k-1}-\lambda_{2}c^{k-1}
\end{array}\right| \neq 0.\]
A straightforward calculation shows that
\[D= \left( \lambda_{1} - \lambda_{2} \right)
\left|\begin{array}{ccccc}
1 & \cdots & 1 & 1 & 1 \\
a_{i_{1}} & \cdots & a_{i_{k-2}} & b & c \\
\vdots &  & \vdots & \vdots & \vdots \\
a_{i_{1}}^{k-1} & \cdots & a_{i_{k-2}}^{k-1} & b^{k-1} & c^{k-1}
\end{array}\right|.
\]
Since $\lambda_{1} \neq \lambda_{2}$ and all \( a_i \) (\( 1 \leq i \leq n-1\)), \( b \),  and \( c \) in \( \mathbb{F}_q \) are distinct,
the determinant is obviously nonzero. Therefore, we have obtained the necessary and sufficient condition for $C_2(b,c)$ to be MDS.
\end{proof}
 The condition (\ref{two column}) is satisfied in the following way: let $H$ be a  multiplicative subgroup of $\mathbb{F}_q^*$ with order not less than $n-1$.  For distinct elements $\mu_i\in H\setminus\{1\}$, choose
    \begin{equation}
    a_i=(b-\mu_i c)/(1-\mu_i),\qquad \lambda_1, \lambda_2\notin H, \qquad \lambda_1\neq \lambda_2.\label{two column a}
    \end{equation}
      Then $(b-a_i)/(c-a_i)=\mu_i$.

    Now for any $1\leq i_1\leq \cdots \leq i_{k-1}\leq n-1$, $l=1,2$ and $s=k-2, k-1$,
    \[\prod\limits_{j=1}^{s}\frac{b-a_{i_j}}{c-a_{i_j}}=\prod\limits_{j=1}^{s}\mu_{i_j}\in H \; \text{and}\; \lambda_l\notin H.\]
    As a result, the inequalities (\ref{two column}) and  (\ref{two column extended}) hold.

   \begin{theorem}\label{two column construction}
   For positive integers $n, k$, suppose that:
   \begin{itemize}
     \item $a_i, \lambda_1, \lambda_2$ are chosen according to condition  (\ref{two column a}),
     \item $C_2(b,c)$ and $C_2(b,c,\infty)$ are the linear codes generated by $G_2(b,c)$ and $G_2(b,c,\infty)$, respectively.
   \end{itemize}
  Then
\begin{itemize}
  \item[(1).] $C_2(b,c)$ and $C_2(b,c,\infty)$ are MDS codes.
  \item[(2).] For $3\leq k\leq \frac{n}{2} $,  the Schur squares codes of both $C_2(b,c)$ and $C_2(b,c,\infty)$ have dimension $2k$. In particular, both $C_2(b,c)$ and $C_2(b,c,\infty)$ are nonequivalent to any RS code or extended RS code.
   \end{itemize}
   \end{theorem}

   \begin{proof}  It remains to show that $C_2(b,c)$ is nonequivalent to any RS code or extended RS code. We only show the case $C_2(b,c)$.
   The Schur product of $C_2(b,c)$  has a generator matrix $G_2(b,c)* G_2(b,c)$ whose rows are
 \[
 \left(a_1^{i+j},\cdots, a_{n-2}^{i+j}, (b^i-\lambda_1 c^i)(b^j-\lambda_1 c^j) ,(b^i-\lambda_2 c^i)(b^j-\lambda_2 c^j)\right).
 \]
If $k\geq 3$, choosing $(i,j)=(0,2), (1,1)$ respectively, we obtain two rows
   \[\left(a_1^{2}, a_2^{2},\cdots, a_{n-1}^{2}, (1-\lambda_1 )(b^2-\lambda_1 c^2), (1-\lambda_2 )(b^2-\lambda_2 c^2)\right)\]
   and
   \[\left(a_1^{2}, a_2^{2},\cdots, a_{n-1}^{2}, (b-\lambda_1 c)^2,  (b-\lambda_2 c)^2\right).\]
   In this case, the difference of the above two vectors is
   \begin{equation}
   \boldsymbol{\alpha}=\left(0,\cdots, 0, \lambda_1(b-c)^2, \lambda_2(b-c)^2\right)\neq {\bf 0}.\label{diff k=3}
    \end{equation}
    Below we choose $(i,j)$ and $(s,r)$ where $i+j=s+r$. Then we obtain two row vectors:

 \[(a_{1}^{i+j}, \cdots, a_{n-2}^{i+j}, (b^{i}-\lambda_{1}c^{i})(b^{j}-\lambda_{1}c^{j}), (b^{i}-\lambda_{2}c^{i})(b^{j}-\lambda_{2}c^{j})) \]
 and
\[ (a_{1}^{s+r}, \cdots, a_{n-2}^{s+r}, (b^{s}-\lambda_{1}c^{s})(b^{r}-\lambda_{1}c^{r}), (b^{s}-\lambda_{2}c^{s})(b^{r}-\lambda_{2}c^{r})).\]
Subtracting these two,

\[ \small
\boldsymbol{\beta} = (0, \cdots, 0, \lambda_{1}(b^{s}c^{r}+b^{r}c^{s}-b^{i}c^{j}-b^{j}c^{i}), \lambda_{2}(b^{s}c^{r}+b^{r}c^{s}-b^{i}c^{j}-b^{j}c^{i})).
\]
Clearly:
\[
\frac{b^{s}c^{r}+b^{r}c^{s}-b^{i}c^{j}-b^{j}c^{i}}{(b-c)^{2}}  \boldsymbol{\alpha} = \boldsymbol{\beta}.
\]
Define
\[\mathcal{L}= \left \{ l(i,j,s,r)\mid i+j=s+r,0\le i,j,s,r\le k-1  \right \}, \]
where $l(i,j,s,r)=\frac{b^{s}c^{r}+b^{r}c^{s}-b^{i}c^{j}-b^{j}c^{i}}{(b-c)^{2}}$.  Then, rewrite  $\mathcal{L}=\left \{ l_1,l_2,\dots ,l_t  \right \}$ with $t=\left\lvert \mathcal{L}\right\rvert$.
Therefore, $C_{2}(b,c) * C_{2}(b,c)$ has a generator matrix

$$\small
\left(
\begin{array}{ccccc}
1  & \cdots & 1 & * & * \\
a_1  & \cdots & a_{n-2} & * & * \\
a_1^2  & \cdots & a_{n-2}^2 & (1-\lambda_1)(b^2-\lambda_1 c^2) & (1-\lambda_2)(b^2-\lambda_2 c^2) \\
\vdots  & \vdots & \vdots & \vdots & \vdots \\
a_1^{2k-2}  & \cdots & a_{n-2}^{2k-2} & * & * \\
a_1^2  & \cdots & a_{n-2}^2 & (b-\lambda_1 c)^2 & (b-\lambda_2 c)^2 \\
\vdots  & \vdots & \vdots & \vdots & \vdots
\end{array}
\right)
$$
which is row equivalent to
$$\small
\left(\begin{array}{ccccccc}
1 & \cdots & 1 & * & \cdots & * \\
a_{1} & \cdots & a_{n-2} & * & \cdots & * \\
a_{1}^{2} & \cdots & a_{n-2}^{2} & (1-\lambda_{1})(b^{2}-\lambda_{1}c^{2}) & \cdots & (1-\lambda_{2})(b^{2}-\lambda_{2}c^{2}) \\
\vdots & \vdots & \vdots & \vdots & \vdots & \vdots \\
a_{1}^{2k-2} & \cdots & a_{n-2}^{2k-2} & * & \cdots & * \\
\multicolumn{3}{c}{} & \boldsymbol{\alpha} \\
\multicolumn{3}{c}{} & l_{1}\boldsymbol{\alpha} \\
\multicolumn{3}{c}{} & \vdots \\
\multicolumn{3}{c}{} & l_{t}\boldsymbol{\alpha}
\end{array}\right).
$$
Clearly, the rank of the upper half of the above matrix is $2k-1$, while the rank of the lower half is 1. Hence, the rank of a generator matrix of $C_{2}(b,c) * C_{2}(b,c)$ is $2k$. Similarly,  $C_{2}(b,c,\infty) * C_{2}(b,c,\infty)$ also has dimension $2k$.  Hence $C_2(b,c)$ and $C_{2}(b,c,\infty)$ are not equivalent to (extended) RS codes.
   \end{proof}
\begin{remark}
  In general, the MDS codes in the above result are not equivalent to previous TGRS codes since their Schur product codes have different dimensions. Here are some examples. Let \(\boldsymbol{a} = (a_1, a_2, \ldots, a_n) \in \mathbb{F}_q^n\) with distinct $a_i$. Then, the code $C_2(b,c)$ and  $C_{2}(b,c,\infty)$ are not equivalent to the following TRS codes:
  \begin{enumerate}
    \item (\!\!\cite{SM2022}) For $3 \leq k < \lfloor  \frac{n}{2}\rfloor-1 $, $l = 2$, $\boldsymbol{t}=(0,1)$, $\boldsymbol{h}=(k-1,k-2)$, and $\boldsymbol{\eta}=(\eta_1,\eta_2)\in(\mathbb{F}_q^*)^2$, the code $\mathcal{C}_{k,n,\boldsymbol{t},\boldsymbol{h}}(\boldsymbol{a},\boldsymbol{1},\boldsymbol{\eta})$ satisfies $\dim((\mathcal{C}_{k,n,\boldsymbol{t},\boldsymbol{h}}(\boldsymbol{a},\boldsymbol{1},\boldsymbol{\eta}))^2)\geq 2k+1$.

    \item (\!\!\cite{SJM2022}) For $3 \leq k < \frac{n}{2}$, $l = 1$, $0 \leq t \leq n-k $, $0 \leq h \leq k$, and $\eta \in \mathbb{F}_q^*$, the code $\mathcal{C}_{k,n,\boldsymbol{t},\boldsymbol{h}}(\boldsymbol{a},\boldsymbol{1},\eta)$ satisfies $\dim((\mathcal{C}_{k,n,\boldsymbol{t},\boldsymbol{h}}(\boldsymbol{a},\boldsymbol{1},\eta))^2)\geq 2k$.

    \item (\!\!\cite{CA2024}) For $4 \leq k < \frac{n}{2}$, $l = 2$, $\boldsymbol{t}=(1,2)$, $\boldsymbol{h}=(k-2,k-2)$, and $\boldsymbol{\eta}=(\eta_1,\eta_2)\in(\mathbb{F}_q^*)^2$, the code $\mathcal{C}_{k,n,\boldsymbol{t},\boldsymbol{h}}(\boldsymbol{a},\boldsymbol{1},\boldsymbol{\eta})$ satisfies $\dim((\mathcal{C}_{k,n,\boldsymbol{t},\boldsymbol{h}}(\boldsymbol{a},\boldsymbol{1},\boldsymbol{\eta}))^2)\geq 2k$.

\end{enumerate}
\end{remark}
   Different subgroups $H$ will produce MDS codes with different maximal length. Combining Theorems \ref{one column construction} and \ref{two column construction}, we obtain a large class of MDS codes that are inequivalent to both classical (extended) RS codes and conventionally defined (extended) TGRS codes. We now systematically outline the detailed construction procedure. For prime power $q$, the construction of two-column TGRS code with parameters $[n,k,n-k+1]$ can be divided into the following five steps:
   \begin{enumerate}
    \item Select a multiplicative subgroup $H$ of $\mathbb{F}_q^*$ and choose distinct elements $b \neq c \in \mathbb{F}_q^*$.

    \item Take a set $\{\mu_1, \mu_2, \dots, \mu_n\} \subseteq  H\setminus\{1\}$, where $1$ denotes the identity element of $H$.

    \item Compute the evaluation points $\{a_1, a_2, \dots, a_n\}$ via:
    $$
    a_i = \frac{b - \mu_i c}{1 - \mu_i}.
    $$

    \item Choose distinct elements $\lambda_1 \neq \lambda_2 \in \mathbb{F}_q^* \setminus H$, with the dimension satisfying $3 \leq k \leq n/2$.

    \item Construct the generator matrices $G_2(b,c)$ and $G_2(b,c,\infty)$.
\end{enumerate}

   \begin{corollary}
   The following MDS codes inequivalent to (extended) RS code can be explicitly proposed.
   \begin{itemize}
     \item[(1)] For $q$ odd prime power, there exists MDS code with parameters $[n, k, n-k+1]_q$ for any $n\leq \frac{q+3}{2}$.
     \item[(2)] For $q=2^{2m}$, there exists MDS code with parameters $[n, k, n-k+1]_q$ for any $n\leq \frac{q+8}{3}$.
   \end{itemize}
   \end{corollary}
   \begin{proof}
   \begin{itemize}
     \item[(1)]  For an odd prime power $q$, let $H$ be the subgroup of squares in $\mathbb{F}_q^*$, and let $\{\mu_i\}$ range over all elements of $H\setminus\{1\}$. Thus, for any $1\leq i_1\leq \cdots \leq i_{k-1}\leq n-1$,
     \[\prod\limits_{j=1}^{k-1}\frac{b-a_{i_j}}{c-a_{i_j}}=\prod\limits_{j=1}^{k-1}\mu_{i_j}\in H.\]
     Since $\lambda_1 \notin H$, then
     \[\prod\limits_{j=1}^{k-1}\frac{b-a_{i_j}}{c-a_{i_j}}\neq \lambda_1.\]Therefore, from Lemma(\ref{one column criterion}), $C_1(b,c)$ and $C_1(b,c,\infty)$ are MDS codes. Similarly, from Lemma(\ref{two column criterion}), we obtain that $C_2(b,c)$ and $C_2(b,c,\infty)$ are MDS codes. Since $\mu_i$ runs through all of $H\setminus\{1\}$, the generator matrix of $C_2(b,c,\infty)$ therefore has $\left\lvert H\setminus\{1\}\right\rvert $ plus the three column vectors we added, totaling $\left\lvert H\setminus\{1\}\right\rvert+3 $ column vectors. Hence the code length of $C_2(b,c,\infty)$ is $\left\lvert H\setminus\{1\}\right\rvert+3=\frac{q-1}{2}-1+3=\frac{q+3}{2}  $. Therefore, we obtain that the code-lengths of $C_1(b,c)$, $C_1(b,c,\infty)$, $C_2(b,c)$, and $C_2(b,c,\infty)$ are at most $\frac{q-1}{2}$,  $\frac{q+1}{2}$, $\frac{q+1}{2}$, and $\frac{q+3}{2}$, respectively.

     \item[(2)]   For $q=2^{2m}$, choosing $H$ as subgroup consisting of all cubics in $\mathbb{F}_q^*$. Let $w$ be a primitive element of $\mathbb{F}_q$ and let $\{\mu_i\}$ take all elements in $\{w^2\}\cup (H \setminus \{1\})$. Then we obtain the set of evaluation points $\{a_1, \dots, a_{\left\lvert H \right\rvert -1}, a_{\left\lvert H \right\rvert}\}$. Take $\lambda_1 \neq \lambda_2$ with $\lambda_1, \lambda_2 \in wH$. Thus, for any $1\leq i_1\leq \cdots \leq i_{k-1}\leq n-1$,
\[\prod\limits_{j=1}^{k-1}\frac{b-a_{i_j}}{c-a_{i_j}}=\prod\limits_{j=1}^{k-1}\mu_{i_j}\in H ~ \text{or}~ w^{2}H.\]
Obviously, $\lambda_1,\lambda_2 \notin H \cup w^2H$, then \[\prod\limits_{j=1}^{k-1}\frac{b-a_{i_j}}{c-a_{i_j}}\neq \lambda_1, \prod\limits_{j=1}^{k-1}\frac{b-a_{i_j}}{c-a_{i_j}}\neq \lambda_2.\]
 Thus, by Lemmas \ref{one column criterion} and \ref{two column criterion},  $C_1(b,c)$, $C_1(b,c,\infty)$, $C_2(b,c)$, and $C_2(b,c,\infty)$ are MDS codes. Similar to the proof of Corollary 1, since $\mu_i$ runs through all elements of $\{w^2\}\cup (H \setminus \{1\})$, the code length of $C_2(b,c,\infty)$ is $\left\lvert \{w^2\}\cup (H \setminus \{1\})\right\rvert+3=\frac{q-1}{3}+3=\frac{q+8}{3} $. Therefore, we obtain that the code-lengths of $C_1(b,c)$, $C_1(b,c,\infty)$, $C_2(b,c)$, and $C_2(b,c,\infty)$ are at most $\frac{q+2}{3}$,  $\frac{q+5}{3}$, $\frac{q+5}{3}$, and $\frac{q+8}{3}$, respectively.
   \end{itemize}
   \end{proof}

For large odd prime power $q$, the code length of systematically constructed $q$-ary TGRS codes in previous studies was bounded by $\frac{q+1}{2}$. For some cases with small parameters, examples with code lengths exceeding $\frac{q+1}{2}$ can be found through Magma computations(\!\!\cite{GO2024}). Different from conventional TGRS codes, our MDS codes can achieve maximal length of $\frac{q+3}{2}$, and we can precisely determine the dimension of the Schur square codes for these MDS codes. More importantly, the construction method proposed in this paper has clear and intuitive characteristics. It shows greater flexibility in choosing $b,c,H, \lambda_1$ and $\lambda_2$. This characteristic gives our method significant advantage in practical applications. Based on the above two results, we now present some examples.

\begin{example}
   For $q=29$, let \[H = \{1, 4, 5, 6, 7, 9, 13, 16, 19, 20, 22, 23, 24, 25\}\] be the subgroup of squares in $\mathbb{F}_{29}^*$.
  Let $ b = 12 $ and $ c = 7 $. Substitute into the formula
$$
a_i = \frac{b - \mu_i c}{1 - \mu_i},
$$
where $ \mu_i \in H $ and $ \mu_i \neq 1 $.
We obtain the evaluation point set \[ T = \{3, 4, 6, 8, 9, 10, 11, 13, 15, 16, 22, 24, 26\} .\]Let $ \lambda_1 = 15 $ and $ \lambda_2 = 21 $. a generator matrix of the linear code $ C_2(12,7,\infty) $ is obtained as follows:
\[G_2(12,7,\infty) =\left(
\begin{array}{cccccccccccccccc}
1 & 1 & 1 & 1 & 1 & 1 & 1 & 1 & 1 & 1 & 1 & 1 & 1 & 15 & 9 & 0 \\
3 & 4 & 6 & 8 & 9 & 10 & 11 & 13 & 15 & 16 & 22 & 24 & 26 & 23 & 10 & 0 \\
9 & 16 & 7 & 6 & 23 & 13 & 5 & 24 & 22 & 24 & 20 & 25 & 9 & 18 & 14 & 0 \\
27 & 6 & 13 & 19 & 4 & 14 & 26 & 22 & 11 & 7 & 5 & 20 & 2 & 5 & 6 & 0 \\
23 & 24 & 20 & 7 & 7 & 24 & 25 & 25 & 20 & 25 & 23 & 16 & 23 & 4 & 11 & 0 \\
11 & 9 & 4 & 27 & 5 & 8 & 14 & 6 & 10 & 23 & 13 & 7 & 18 & 4 & 24 & 0 \\
4 & 7 & 24 & 13 & 16 & 22 & 9 & 20 & 5 & 20 & 25 & 23 & 4 & 1 & 25 & 1
\end{array}
\right).
\]
By \textsc{Magma}, $ C_2(12,7,\infty) $ is an MDS code with parameters $[16, 7, 10]$. Its Schur square code has parameters $[16, 14, 2]$, which implies that $ C_2(12,7,\infty) $ is not equivalent to any (\textit{extended})  GRS code.

\end{example}

\begin{example}
   For $q=3^3$, let $w$ be a primitive element of $\mathbb{F}_{3^3}$, and let
   \[
   H= \left\{ 1,\, w^{22},\, w^{12},\, w^{2},\, w^{24},\, w^{14},\, w^{4},\, w^{16},\, w^{6},\, w^{18},\, w^{8},\, w^{20},\, w^{10} \right\}\]
   be the subgroup of squares in $\mathbb{F}_{3^3}^*$.
    Let $ b = w^{7} $ and $ c = w^{11} $. Then, by the same method, we derive the evaluation point set
\[
\begin{split}
T= \left\{ w^{22},\, w^{12},\, w^{2},\, w^{24},\, w^{14},\, w^{4},\, w^{16},\, w^{6},\, w^{18},\, w^{8},\, w^{20},\, w^{10} \right\}.
\end{split}
\]
Let $ \lambda_1 = w^{15} $ and $ \lambda_2 = w^{21} $. A generator matrix of the linear code $ C_2(w^7,w^{11},\infty) $ is obtained as follows:

\[
G_2(w^7,w^{11},\infty) =
\left(
\begin{array}{ccccccccccccccc}
1 & 1 & 1 & 1 & 1 & 1 & 1 & 1 & 1 & 1 & 1 & 1 & w^{21} & w^{15} & 0 \\
0 & 1 & w^{5} & w^{22} & w^{15} & w^{25} & w^{10} & w^{8} & w^{19} & w^{3} & 2 & w^{18} & w^{18} & w^{9} & 0 \\
0 & 1 & w^{10} & w^{18} & w^{4} & w^{24} & w^{20} & w^{16} & w^{12} & w^{6} & 1 & w^{10} & w^{20} & w^{10} & 0 \\
0 & 1 & w^{15} & w^{14} & w^{19} & w^{23} & w^{4} & w^{24}  & w^{5} & w^{9} & 2 & w^{2} & w^{11} & 1 & 0 \\
0 & 1 & w^{20} & w^{10} & w^{8} & w^{22} & w^{14} & w^{6}  & w^{24} & w^{12} & 1 & w^{20} & w^{9} & w^{21} & 0 \\
0 & 1 & w^{25} & w^{6} & w^{23} & w^{21} & w^{24} & w^{14} & w^{17} & w^{15} & 2 & w^{12} & w^{23} & w^{4} & 0 \\
0 & 1 & w^{4} & w^{2} & w^{12} & w^{20} & w^{8} & w^{22} & w^{10} & w^{18} & 1 & w^{4} & w^{3} & w & 1
\end{array}
\right).
\]
By \textsc{Magma}, $ C_2(w^7,w^{11},\infty) $ is an MDS code with parameters $[15, 7, 9]$. Its Schur square code has parameters $[15, 14, 1]$, which implies that $ C_2(w^7,w^{11},\infty) $ is not equivalent to any (\textit{extended}) GRS code.
\end{example}

\begin{example}
  For $q=2^6$, let $w$ be a primitive element of  $\mathbb{F}_{2^6}$, and let
   \[
\begin{split}
H := \left\{ 1,\, w^{30},\, w^{60},\, w^{3},\, w^{33},\, w^{6},\, w^{36},\, w^{9},\, w^{39},\, w^{12},\, w^{42},\, w^{15},\, w^{45},\, w^{18},\, w^{48},\, w^{21},\, w^{51},\, w^{24},\, w^{54},\, w^{27},\, w^{57} \right\}
\end{split}
\] be the subgroup consisting of all cubics in $\mathbb{F}_{2^6}$.
    Let $ b = w^{10} $ and $ c = w^{21} $. Subsequently, we take
\[
H^{\prime} := \left\{ w^{30},\, w^{60},\, w^{3},\, w^{33},\, w^{6},\, w^{36},\, w^{9},\, w^{39},\, w^{12} \right\},
\]
and add  $w^2 \in w^{2}H$. By computation, we obtain its valuation  set
\[
T:=\left\{w^{27}, w^{29},\, w^{7},\, w^{13},\, w,\, w^{23},\, w^{26},\, w^{40},\, w^{46},\, w^{32} \right\}.
\]
Let $ \lambda_1 = w^{13} $ and $ \lambda_2 = w^{25} $ with $\lambda_1, \lambda_2 \in wH$. A generator matrix of the linear code $ C_2(w^{10},w^{21},\infty) $ is obtained as follows:
\[
G_2(w^{10},w^{21},\infty) =
\left(
\begin{array}{ccccccccccccc}
1 & 1 & 1 & 1 & 1 & 1 & 1 & 1 & 1 & 1 & w^3 & w^{58} & 0 \\
w^{27} & w^{29} & w^7 & w^{13} & w & w^{23} & w^{26} & w^{40} &w^{46} & w^{32} & w^{51} & w^{55} & 0 \\
w^{54} & w^{58} & w^{14} & w^{26} & w^2 & w^{46} & w^{52} & w^{17} & w^{29} & w & w^{24} & w^{18} & 0 \\
w^{18} & w^{24} & w^{21} & w^{39} & w^3 & w^6 & w^{15} & w^{57} & w^{12} & w^{33} & w^{52} & w^{55} & 0 \\
w^{45} & w^{53} & w^{28} & w^{52} & w^4 & w^{29} & w^{41} & w^{34} & w^{58} & w^2 & w^{60} & w^3 & 1
\end{array}
\right).
\]
By \textsc{Magma}, $ C_2(w^{10},w^{21},\infty) $ is an MDS code with parameters $[13, 5, 9]$. Its Schur square code has parameters $[13, 10, 2]$, which implies that $ C_2(w^{10},w^{21},\infty) $ is not equivalent to any (extended)  GRS code.
\end{example}

\section{Dual Codes of Column TRS Codes}

In this section, we present the dual code of $C_2(b,c)$ and $C_2(b,c,\infty)$. As is well known, a parity-check matrix of $C_2(b,c)$ also serves as a generator matrix of its dual code. For $1\leq l \leq n-1$, define
\[\Delta_l(x) = {\prod \limits_{\substack{i=1 \\ i\neq l}}^{k} (a_i - x)}.\]
We begin by establishing the following result:

\begin{theorem}\label{Dual column TRS codes}
    \begin{itemize}
        \item [\rm (1)]A parity-check martix of $C_2(b,c)$ is
\[ H_2(b,c)=
\left(
\begin{array}{ccccccc}
\frac{\Delta_1(a_{k+1})}{\Delta_1(a_{1})} & \cdots & \frac{\Delta_k(a_{k+1})}{\Delta_k(a_{k})} & -1  & \cdots & 0 & 0 \\
\frac{\Delta_1(a_{k+2})}{\Delta_1(a_{1})} & \cdots & \frac{\Delta_1(a_{k+2})}{\Delta_k(a_{k})} & 0  & \cdots & 0 & 0 \\
\vdots & \vdots & \vdots & \vdots  & \vdots & \vdots & \vdots \\
\frac{\Delta_1(b)-\lambda_1\Delta_1(c)}{\Delta_1(a_{1})} & \cdots & \frac{\Delta_k(b)-\lambda_1\Delta_k(c)}{\Delta_k(a_{k})} & 0  & \cdots & -1 & 0 \\
\frac{\Delta_1(b)-\lambda_2\Delta_1(c)}{\Delta_1(a_{1})} & \cdots & \frac{\Delta_k(b)-\lambda_2\Delta_k(c)}{\Delta_k(a_{k})} & 0  & \cdots & 0 & -1
\end{array}
\right).
\]
        \item [\rm (2)]For $k\leq \frac{n}{2}$, a parity-check martix of $C_2(b,c,\infty)$ is
\[ H_2(b,c,\infty)=
\left(
\begin{array}{cccccccc}
\frac{\Delta_1(a_{k+1})}{\Delta_1(a_{1})} & \cdots & \frac{\Delta_k(a_{k+1})}{\Delta_k(a_{k})} & -1 & \cdots & 0 & 0  &0\\
\frac{\Delta_1(a_{k+2})}{\Delta_1(a_{1})} & \cdots & \frac{\Delta_k(a_{k+2})}{\Delta_k(a_{k})} & 0 & \cdots & 0 & 0 &0\\
\vdots & \vdots & \vdots & \vdots & \vdots & \vdots & \vdots \\
\frac{\Delta_1(b)-\lambda_1\Delta_1(c)}{\Delta_1(a_{1})} & \cdots & \frac{\Delta_k(b)-\lambda_1\Delta_k(c)}{\Delta_k(a_{k})} & 0 & \cdots & -1 & 0 &0\\
\frac{\Delta_1(b)-\lambda_2\Delta_1(c)}{\Delta_1(a_{1})} & \cdots & \frac{\Delta_k(b)-\lambda_2\Delta_k(c)}{\Delta_k(a_{k})} & 0 & \cdots & 0 & -1 &0\\
\frac{(-1)^{k+1}}{\Delta_1(a_{1})} & \cdots & \frac{(-1)^{k+1}}{\Delta_k(a_{k})} & 0 & \cdots & 0 &0& -1
\end{array}
\right).
\]
    \end{itemize}
\end{theorem}

\begin{proof}
\begin{itemize}
    \item [\rm (1)]
    Recall that a generator matrix $C_2(b,c)$ is
\[
G_2(b,c) =
\left(
      \begin{array}{ccccc}
        1  & \cdots & 1 & 1-\lambda_1 & 1-\lambda_2 \\
        a_1  & \cdots & a_{n-1}  & b-\lambda_1 c  & b-\lambda_2 c \\
        a_1^2  & \cdots & a_{n-1}^2 & b^2-\lambda_1 c^2  & b^2-\lambda_2 c^2\\
        \vdots  & \vdots & \vdots & \vdots &\vdots \\
        a_1^{k-1}  & \cdots & a_{n-1}^{k-1} & b^{k-1}-\lambda_1 c^{k-1} & b^{k-1}-\lambda_2  c^{k-1}\\
      \end{array}
    \right).
\]
    Suppose, $v=(v_1,v_2,\cdots,v_{n+1})$ satisfies the equation $vG_2^T(b,c)=0$. We obtain
   $$
\begin{cases}
v_1  + \cdots + v_{n-1} + (1-\lambda_1)v_n + (1-\lambda_2)v_{n+1} = 0, \\[5pt]
a_1 v_1  + \cdots + a_{n-1} v_{n-1} + (b - \lambda_1 c)v_n + (b - \lambda_2 c)v_{n+1} = 0, \\[5pt]
~\cdots \quad \cdots \quad \cdots \\[5pt]
a_1^{k-1} v_1  + \cdots + a_{n-1}^{k-1} v_{n-1} + (b^{k-1} - \lambda_1 c^{k-1})v_n + (b^{k-1} - \lambda_2 c^{k-1})v_{n+1} = 0.
\end{cases}
$$
    The coefficient matrix formed by the first $k$ unknowns is of full rank. Therefore, given a set of free variables as follows
    $$
    \begin{cases}
        (v_{k+1},v_{k+2},\cdots,v_{n+1}) = (-1,0,\cdots,0), \\[5pt]
        (v_{k+1},v_{k+2},\cdots,v_{n+1}) = (0,-1,\cdots,0), \\[5pt]
        ~\cdots \quad \cdots \quad \cdots \\[5pt]
        (v_{k+1},v_{k+2},\cdots,v_{n+1}) = (0,0,\cdots,-1).
    \end{cases}
    $$
    It follows that the solution space $V={\rm Span}(\boldsymbol{x}_{k+1},\cdots,\boldsymbol{x}_{n+1})$ of equations can be obtained using Cramer's rule, where
 {\small       $$
    \begin{cases}
        \boldsymbol{x}_{k+1} = (\frac{\Delta_1(a_{k+1})}{\Delta_1(a_{1})} , \cdots , \frac{\Delta_k(a_{k+1})}{\Delta_k(a_{k})} , -1  , \cdots  ,0 , 0), \\[5pt]
        ~\cdots \quad \cdots \quad \cdots \\[5pt]
        \boldsymbol{x}_{n} = (\frac{\Delta_1(b)-\lambda_1\Delta_1(c)}{\Delta_1(a_{1})}  , \cdots , \frac{\Delta_k(b)-\lambda_1\Delta_k(c)}{\Delta_k(a_{k})} , 0 , \cdots , -1 , 0), \\[5pt]
        \boldsymbol{x}_{n+1} = (\frac{\Delta_1(b)-\lambda_2\Delta_1(c)}{\Delta_1(a_{1})} ,  \cdots , \frac{\Delta_k(b)-\lambda_2\Delta_k(c)}{\Delta_k(a_{k})} , 0 , \cdots , 0 , -1).
    \end{cases}
    $$}
Therefore, we obtain the parity-check matrix $H_2(b,c)$ of $C_2(b,c)$, which is exactly a generator matrix of its dual code.
    \item [\rm (2)]
    Similarly, consider the equation
    $$
    \begin{cases}
    v_1  + \cdots + v_{n-1} + (1-\lambda_1)v_n + (1-\lambda_2)v_{n+1} = 0, \\[5pt]
    a_1 v_1  + \cdots + a_{n-1} v_{n-1} + (b - \lambda_1 c)v_n + (b - \lambda_2 c)v_{n+1} = 0, \\[5pt]
    ~\cdots \quad \cdots \quad \cdots \\[5pt]
    a_1^{k-2} v_1  + \cdots + a_{n-1}^{k-2} v_{n-1} + (b^{k-2} - \lambda_1 c^{k-2})v_n +  (b^{k-2} - \lambda_2 c^{k-2})v_{n+1} = 0, \\[5pt]
    a_1^{k-1} v_1  + \cdots + a_{n-1}^{k-1} v_{n-1} + (b^{k-1} - \lambda_1 c^{k-1})v_n +  (b^{k-1} - \lambda_2 c^{k-1})v_{n+1}  + v_{n+2}= 0.
    \end{cases}
    $$
    The coefficient matrix formed by the first $k$ unknowns is of full rank. Therefore, given a set of free variables as follows
    $$
    \begin{cases}
        (v_{k+1},v_{k+2},\cdots,v_{n+1}) = (-1,0,\cdots,0), \\[5pt]
        (v_{k+1},v_{k+2},\cdots,v_{n+1}) = (0,-1,\cdots,0), \\[5pt]
        ~\cdots \quad \cdots \quad \cdots \\[5pt]
        (v_{k+1},v_{k+2},\cdots,v_{n+1}) = (0,0,\cdots,-1).
    \end{cases}
    $$
    It follows that the solution space $V={\rm Span}(\boldsymbol{x}_{k+1},\cdots,\boldsymbol{x}_{n+1},\boldsymbol{x}_{n+2})$ of equations can be obtained using Cramer's rule, where
      $$
    \begin{cases}
        \boldsymbol{x}_{k+1} = \left(\frac{\Delta_1(a_{k+1})}{\Delta_1(a_{1})}  , \cdots , \frac{\Delta_k(a_{k+1})}{\Delta_k(a_{k})} , -1  , \cdots,0  ,0 , 0\right), \\[5pt]
        ~\cdots \quad \cdots \quad \cdots \\[5pt]
        \boldsymbol{x}_{n} = \left(\frac{\Delta_1(b)-\lambda_1\Delta_1(c)}{\Delta_1(a_{1})}  , \cdots , \frac{\Delta_k(b)-\lambda_1\Delta_k(c)}{\Delta_k(a_{k})} , 0 , \cdots , -1 , 0,0\right), \\[5pt]
        \boldsymbol{x}_{n+1} = \left(\frac{\Delta_1(b)-\lambda_2\Delta_1(c)}{\Delta_1(a_{1})}  , \cdots , \frac{\Delta_k(b)-\lambda_2\Delta_k(c)}{\Delta_k(a_{k})} , 0 , \cdots , 0 , -1,0\right),\\[5pt]
        \boldsymbol{x}_{n+2} = \left(\frac{(-1)^{k+1}}{\Delta_1(a_{1})} ,  \cdots , \frac{(-1)^{k+1}}{\Delta_k(a_{k})} , 0  , \cdots , 0 , 0 , -1\right).
    \end{cases}
    $$
   Hence, we obtain the parity-check matrix $H_2(b,c,\infty)$ of $C_2(b,c,\infty)$, which is exactly a generator matrix of its dual code.
\end{itemize}

\end{proof}

According to Theorem \ref{Dual column TRS codes}, we can calculate parity-check matrices for $C_2(12,7,\infty)$, $C_2(w^7,w^{11},\infty)$, and $C_2(w^{10},w^{21},\infty)$ in {\bf Examples} 1-3 as follows:
\[
\begin{array}{ll}
H_2(12,7,\infty) =\left(
\begin{array}{cccccccccccccccc}
16 & 24 & 15 & 24 & 18 & 12 & 8 & -1 & 0 & 0 & 0 & 0 & 0 & 0 & 0 & 0 \\
2 & 4 & 9 & 15 & 25 & 15 & 18 & 0 & -1 & 0 & 0 & 0 & 0 & 0 & 0 & 0 \\
11 & 15 & 16 & 27 & 21 & 5 & 22 & 0 & 0 & -1 & 0 & 0 & 0 & 0 & 0 & 0 \\
4 & 26 & 26 & 2 & 12 & 21 & 26 & 0 & 0 & 0 & -1 & 0 & 0 & 0 & 0 & 0 \\
17 & 3 & 18 & 19 & 11 & 9 & 11 & 0 & 0 & 0 & 0 & -1 & 0 & 0 & 0 & 0 \\
19 & 5 & 21 & 17 & 5 & 9 & 12 & 0 & 0 & 0 & 0 & 0 & -1 & 0 & 0 & 0 \\
17 & 15 & 8 & 6 & 8 & 1 & 18 & 0 & 0 & 0 & 0 & 0 & 0 & -1 & 0 & 0 \\
2 & 17 & 23 & 22 & 7 & 20 & 5 & 0 & 0 & 0 & 0 & 0 & 0 & 0 & -1 & 0 \\
24 & 15 & 23 & 18 & 5 & 17 & 14 & 0 & 0 & 0 & 0 & 0 & 0 & 0 & 0 & -1
\end{array}
\right), &\\[4mm]
H_2(w^7,w^{11},\infty) = \left(
\begin{array}{ccccccccccccccc}
w^{25} & w^{2} & w^{24} & w^{15} & w^{21} & w^{25} & w^{17} & -1 & 0 & 0 & 0 & 0 & 0 & 0 & 0 \\
w^{12} & w & w^{9} & w^{15} & w^{10} & w^{7} & w^{11} & 0 & -1 & 0 & 0 & 0 & 0 & 0 & 0 \\
w^{25} & 2 & w^{5} & 2 & w^{18} & w^{14} & w^{11} & 0 & 0 & -1 & 0 & 0 & 0 & 0 & 0 \\
w^{6} & 2 & w^{20} & w^{17} & w^{8} & w^{16} & w^{17} & 0 & 0 & 0 & -1 & 0 & 0 & 0 & 0 \\
1 & w^{8} & w^{18} & w^{5} & w^{21} & 2 & 1 & 0 & 0 & 0 & 0 & -1 & 0 & 0 & 0 \\
w^{21} & w^{8} & w^{17} & w^{16} & w^{8} & w^{3} & w^{22} & 0 & 0 & 0 & 0 & 0 & -1 & 0 & 0 \\
w^{22} & w^{7} & w^{6} & w^{6} & w^{6} & w^{21} & w^{16} & 0 & 0 & 0 & 0 & 0 & 0 & -1 & 0 \\
w & w^{21} & w^{6} & 1 & w^{2} & w^{19} & w^{18} & 0 & 0 & 0 & 0 & 0 & 0 & 0 & -1
\end{array}
\right), &\\[4mm]
H_2(w^{10},w^{21},\infty) =
\left(
\begin{array}{ccccccccccccc}
w^{14} & w^{10} & w^{50} & w^{33} & w^{18} & 1 & 0 & 0 & 0 & 0 & 0 & 0 & 0 \\
w^{56} & w^{23} & w^{19} & w^{30} & w^{9} & 0 & 1 & 0 & 0 & 0 & 0 & 0 & 0 \\
w^{24} & w^{52} & w^{8} & w^{3} & w^{29} & 0 & 0 & 1 & 0 & 0 & 0 & 0 & 0 \\
w^{4} & w^{61} & w^{54} & w^{62} & w^{35} & 0 & 0 & 0 & 1 & 0 & 0 & 0 & 0 \\
w^{52} & w^{54} & w^{43} & w^{19} & w^{42} & 0 & 0 & 0 & 0 & 1 & 0 & 0 & 0 \\
w^{19} & w^{18} & w^{51} & w^{29} & w^{59} & 0 & 0 & 0 & 0 & 0 & 1 & 0 & 0 \\
w^{11} & w^{44} & w^{38} & w^{19} & w^{31} & 0 & 0 & 0 & 0 & 0 & 0 & 1 & 0 \\
w^{13} & 1 & w^{12} & w^{47} & w^{6} & 0 & 0 & 0 & 0 & 0 & 0 & 0 & 1
\end{array}
\right). &
\end{array}
\]
By \textsc{Magma}, parity-check matrices $ H_2(12,7,\infty) $, $ H_2(w^7,w^{11},\infty) $, and $ H_2(w^{10},w^{21},\infty) $ correspond to the codes $ C_2(12,7,\infty) $, $ C_2(w^7,w^{11},\infty) $, and $ C_2(w^{10},w^{21},\infty) $ in {\bf Examples} 1-3 respectively.
\section{CONCLUSION}

In this paper, we established  some conditions for column TRS codes to be MDS. We proved that the dimension of their Schur square codes  is exactly $2k$, demonstrating that they are inequivalent to GRS codes and conventional TGRS codes, where the latter had been shown to be inequivalent to GRS codes. For large odd prime powers $q$, unlike the systematically constructed TGRS codes whose lengths were previously bounded by $\frac{q+1}{2}$, these MDS codes in our constructions have lengths up to $\frac{q+3}{2}$. Precisely, our construction was presented in five steps. Another advantage of this method is its flexibility in choosing evaluation sets and parameters $\lambda_1, \lambda_2$.  In this sense, this approach provided a new method for MDS code design. Finally, we presented the dual codes of these column TGRS codes.

For future work, we aim to investigate whether adding more columns could construct longer MDS column TGRS codes that remain inequivalent to GRS codes.



\begin{thebibliography}{1}
\bibliographystyle{IEEEtran}

\bibitem{BS2018}
P. Beelen, M. Bossert, S. Puchinger, and J. Rosenkilde,
``Structural properties of twisted Reed-Solomon codes with applications to cryptography,''
\textit{Proc. IEEE Int. Symp. Inform. Theory (ISIT)},
pp. 946-950, 2018.

\bibitem{BO2015}
D. Bartoli, M. Giulietti, and I. Platoni,
``On the covering radius of MDS codes,''
\textit{IEEE Trans. Inf. Theory},
vol. 61, no. 2, pp. 801--811, Feb. 2015.

\bibitem{BT2017}
P. Beelen, S. Puchinger, and J. Rosenkilde n$\mathrm{\acute{e}}$ Nielsen,
``Twisted Reed-Solomon codes,''
\textit{Proc. IEEE Int. Symp. Inform. Theory (ISIT)},
pp. 336-340, 2017.

\bibitem{CO2024}
W. Cheng,
``On parity-check matrices of twisted generalized Reed-Solomon codes,''
\textit{IEEE Trans. Inform. Theory},
vol. 70, no. 5, pp. 3213-3225, May 2024.

\bibitem{CO2008}
R. Cramer, V. Daza, I. Gracia, J. J. Urroz, G. Leander and J.
Marti-Farre,
``On codes, matroids, and secure multiparty computation from linear secret-sharing schemes,''
\textit{IEEE Trans. Inf. Theory},
vol. 54, no. 6, pp. 2644--2657, Jun. 2008.

\bibitem{CN2018}
A. Chowdhury and A. Vardy,
``New constructions of MDS codes with asymptotically optimal repair,''
in \textit{Proc. IEEE Int. Symp. Inf. Theory (ISIT)},
pp. 1944--1948, 2018.

\bibitem{CA2003}
W. E. Cherowitzo, C. M. O'Keefe, and T. Penttila,
``A unified construction of finite geometries associated with $q$-clans in characteristic 2,''
\textit{Adv. Geom.},
vol. 3, no. 1, pp. 1--21, 2003.

\bibitem{CE2018}
C. Carlet, S. Mesnager, C. Tang, and Y. Qi,
``Euclidean and Hermitian LCD MDS codes,''
\textit{Des., Codes Cryptogr.},
vol. 86, pp. 2605--2618, Nov. 2018.

\bibitem{CD2014}
A. Couvreur, P. Gaborit, V. Gauthier-Uma$\mathrm{\tilde{n}}$a, A. Otmani and J-P. Tillich,
``Distinguisher-based attacks on public-key cryptosystems using Reed--Solomon codes,''
\textit{Des. Codes Cryptogr.},
vol. 73, pp. 641--666, 2014.

\bibitem{YN2025}
Y. Ding and S. Zhu,
``New self-dual codes from TGRS codes with general $l$ twists,''
\textit{Adv. Math. Commun.},
vol. 19, no. 2, pp. 662-675, 2025.

\bibitem{ET2011}
M. F. Ezerman, M. Grassl, and P. Sol\'e,
``The weights in MDS codes,''
\textit{IEEE Trans. Inf. Theory},
vol. 57, no. 1, pp. 392--396, Jan. 2011.

\bibitem{FD2024}
W. Fang and J. Xu,
``Deep holes of twisted Reed-Solomon codes,''
\textit{Proc. IEEE Int. Symp. Inform. Theory (ISIT)},
pp. 488-493, 2024.

\bibitem{GA2011}
D. G. Glynn,
``A condition for ARCS and MDS codes,''
\textit{Des. Codes Cryptogr.},
vol. 58, no. 2, pp. 215--218, Feb. 2011.

\bibitem{GD2023}
G. Guo, R. Li, Y. Liu, and H. Song,
``Duality of generalized twisted Reed-Solomon codes and Hermitian self-dual MDS or NMDS codes,''
\textit{Cryptogr. Commun.},
vol. 15, pp. 383-395, 2023.

\bibitem{GN2008}
T. A. Gulliver, J. L. Kim, and Y. Lee,
``New MDS or near-MDS self-dual codes,''
\textit{IEEE Trans. Inf. Theory},
vol. 54, no. 9, pp. 4354--4360, Sept. 2008.

\bibitem{GO2024}
H. Gu and J. Zhang,
``On twisted generalized Reed-Solomon codes with $l$ twists,''
\textit{IEEE Trans. Inform. Theory},
vol. 70, no. 1, pp. 145-153, Jan. 2024.

\bibitem{HM2023}
D. Huang, Q. Yue, and Y. Niu,
``MDS or NMDS LCD codes from twisted Reed-Solomon codes,''
\textit{Cryptogr. Commun.},
vol. 15, pp. 221-237, 2023.

\bibitem{HM2021}
D. Huang, Q. Yue, Y. Niu, and X. Li,
``MDS or NMDS self-dual codes from twisted generalized Reed-Solomon codes,''
\textit{Designs, Codes Cryptogr.},
vol. 89, no. 9, pp. 2195-2209, 2021.

\bibitem{KO2015}
J. I. Kokkala, D. S. Krotov, and R. J. Patric,
``On the classification of MDS codes,''
\textit{IEEE Trans. Inf. Theory},
vol. 61, no. 12, pp. 6485--6492, Dec. 2015.

\bibitem{LC2019}
K. Lebed, H. Liu, and J. Luo,
``Construction of MDS self-dual codes over finite fields,''
\textit{Finite Fields Appl.},
vol. 59, pp. 199--207, Sept. 2019.

\bibitem{LC2020}
J. Lavauzelle and J. Renner,
``Cryptanalysis of a system based on twisted Reed-Solomon codes,''
\textit{Designs, Codes Cryptogr.},
vol. 88, no. 7, pp. 1285-1300, 2020.

\bibitem{LC2025}
Y. Li, S. Zhu, and Z. Sun,
``Covering radii and deep holes of two classes of extended twisted GRS codes and their applications,''
\textit{IEEE Trans. Inform. Theory},
vol. 71, no. 5, pp. 3516-3530, May 2025.

\bibitem{PI1995}
S. E. Payne, T. Penttila, and I. Pinneri,
``Isomorphisms between Subiaco $q$-clan geometries,''
\textit{Bull. Belg. Math. Soc. Simon Stevin},
vol. 2, no. 2, pp. 197--222, 1995.

\bibitem{SM2024}
H. Singh and K. C. Meena,
``MDS multi-twisted Reed-Solomon codes with small dimensional hull,''
\textit{Cryptogr. Commun.},
vol. 16, pp. 557-578, 2024.

\bibitem{SM2022}
J. Sui, Q. Yue, X. Li, and D. Huang,
``MDS, Near-MDS or 2-MDS self-dual codes via twisted generalized Reed-Solomon codes,''
\textit{IEEE Trans. Inform. Theory},
vol. 68, no. 12, pp. 7832-7841, 2022.

\bibitem{SD2025}
H. Sun, Q. Yue, X. Jia, and C. Li,
``Decoding algorithms of twisted GRS codes and twisted Goppa codes,''
\textit{IEEE Trans. Inform. Theory},
vol. 71, no. 2, pp. 1018-1027, 2025.

\bibitem{SJM2022}
J. Sui, X. Zhu, and X. Shi,
``MDS and near-MDS codes via twisted Reed-Solomon codes,''
\textit{Designs, Codes Cryptogr.},
vol. 90, pp. 1937-1958, 2022.

\bibitem{WT2021}
Y. Wu,
``Twisted Reed-Solomon codes with one-dimensional hull,''
\textit{IEEE Commun. Lett.},
vol. 25, no. 2, pp. 383-386, 2021.

\bibitem{ZA2022}
J. Zhang, Z. Zhou, and C. Tang,
``A class of twisted generalized Reed-Solomon codes,''
\textit{Designs, Codes Cryptogr.},
vol. 90, pp. 1649-1658, 2022.

\bibitem{ZR2025}
C. Zhao, W. Ma, T. Yan, and Y. Sun,
``Research on the construction of maximum distance separable codes via arbitrary twisted generalized Reed-Solomon codes,''
\textit{IEEE Trans. Inform. Theory}, 
vol. 71, no. 7, pp. 5130-5143, 2025.

\bibitem{ZN2025}
Y. Zhi and S. Zhu,
``New MDS codes of non-GRS type and NMDS codes,''
\textit{Discrete Math.},
vol. 348, no. 5, 114436, 2025.

\bibitem{CA2024}
C. Zhu and Q. Liao,
``A class of double-twisted generalized Reed-Solomon codes,''
\textit{Finite Fields Appl.},
vol. 95, 102395, 2024.

\bibitem{CT2024}
C. Zhu and Q. Liao,
``The (+)-extended twisted generalized Reed-Solomon code,''
\textit{Discrete Math.},
vol. 347, no. 2, 113749, 2024.

\end{thebibliography}
%

\end{document}